\documentclass[12pt,nodisplayskipstretch]{article}

\usepackage{latexsym, amsmath, amssymb, amsfonts,natbib, amsthm, setspace}

\usepackage{geometry}
\usepackage{graphicx,color}
\usepackage{tikz}
\usepackage{verbatim}
\usepackage{hyperref}
\usepackage{subcaption}
\usepackage{enumerate}

\def\ep{\epsilon}
\def\al{\alpha}

\def\epsilon{\varepsilon}

\newcommand{\abs}[1]{| #1 |}

\newtheorem{theorem}{Theorem}

\newtheorem{remark}{Remark}

\newtheorem{corollary}{Corollary}

\newtheorem{lemma}{Lemma}
\newtheorem{proposition}{Proposition}

\newcommand{\beq}{\begin{eqnarray*}}
\newcommand{\eeq}{\end{eqnarray*}}
\newcommand{\beqq}{\begin{eqnarray}}
\newcommand{\eeqq}{\end{eqnarray}}

\newcommand{\bequa}{\begin{equation}}
\newcommand{\eequa}{\end{equation}}

\newcommand{\bit}{\begin{itemize}}
\newcommand{\eit}{\end{itemize}}

\newcommand{\tc}{\textsc}

\title{Did Harold Zuercher Have Time-Separable Preferences?\thanks{
Lu: Department of Economics, University of California, Los Angeles (email: \href{mailto:jay@econ.ucla.edu}{jay@econ.ucla.edu}); Luo: Department of Economics, University of Toronto (email: \href{mailto:yao.luo@utoronto.ca}{yao.luo@utoronto.ca}); Saito: Division of the Humanities and Social Sciences, California Institute of Technology (email: \href{mailto:saito@caltech.edu}{saito@caltech.edu}); Xin: Division of the Humanities and Social Sciences, California Institute of Technology (email: \href{mailto:yixin@caltech.edu}{yixin@caltech.edu}).
Financial support from the NSF under awards SES-1919263 (Saito) and SES-1919275 (Lu) are gratefully acknowledged.
We thank seminar and conference participants at Caltech, Georgetown University, Texas A\&M, University of Arizona, UC Berkeley, UC Irvine, University of Colorado Boulder, and NASMES (2023). We appreciate the valuable discussions we had with Fedor Iskhakov, Peter Klibanoff, Robert Miller, Marciano Siniscalchi, Costis Skiadas, Tomasz Strazalecki and John Rust. In particular, Mitsuru Igami read an earlier version of the manuscript and offered helpful comments. 
}}

\author{Jay Lu, Yao Luo, Kota Saito, and Yi Xin}

\date{\today}

\begin{document}

\maketitle
 
\begin{abstract}
    \singlespacing
    This paper proposes an empirical model of dynamic discrete choice to allow for non-separable time preferences, generalizing the well-known \cite{rust1987optimal} model. Under weak conditions, we show the existence of value functions and hence well-defined optimal choices. 
    We construct a contraction mapping of the value function and propose an estimation method similar to Rust's nested fixed point algorithm.  
    Finally, we apply the framework to the bus engine replacement data. 
    We improve the fit of the data with our general model and reject the null hypothesis that Harold Zuercher has separable time preferences. 
    Misspecifying an agent's preference as time-separable when it is not leads to biased inferences about structure parameters (such as the agent's risk attitudes) and misleading policy recommendations.
    
\end{abstract}

\vspace{.5cm}


\newpage 
\section{Introduction}\label{sec:introduction}

It is well-known that in the standard model of expected utility with separable time preferences, an agent's risk preference and her preference for intertemporal substitution are both captured by the curvature of her Bernoulli utility function. As a result, greater risk aversion is correlated with greater complementarity of consumption across different time periods. Risk and time however are two distinct phenomena and it would be natural to have different preferences over the two. Furthermore, time-separable models are widely recognized for imposing an unrealistic constraint on agents' preferences regarding the timing of uncertainty resolution. As \cite{rust1994structural} pointed out,  time-separable ``expected-utility models imply that agents are indifferent about the timing of the resolution of uncertain events, whereas human decision-makers seem to have definite preferences over the time at which uncertainty is resolved.''

Models with non-separable time preferences such as \cite{epsein1989sub} allow for a clean separation of the two: the elasticity of intertemporal substitution is independent of risk aversion. Moreover, the models allow agents to have nontrivial preferences regarding the timing of the resolution of uncertainty, as suggested by the experimental literature (see, e.g., \citet{nielsen2020preferences}, \citet{meissner2022measuring}).
These feature makes them popular in fields such as macroeconomics and finance, where they have been used to explain phenomena such as the equity premium puzzle (see, e.g., \citet{mehra1985equity} and \cite{bansalyaron2004}). 


However, current empirical research has yet to incorporate non-separable preferences within the dynamic discrete choice (DDC) framework.  Misspecifying an agent's preference as time-separable when it is not can result in biased inferences about structural primitives and misleading policy recommendations. 
Our paper aims to bridge this gap. We begin by outlining a set of theoretical results crucial for empirical analysis in the DDC framework. Following this, we introduce an empirical model of dynamic discrete choice that incorporates non-separable time preferences, extending the well-known bus engine replacement model analyzed in \cite{rust1987optimal}. Our model includes the standard time-separable model, commonly employed in the DDC literature, as a special case. This is an important and empirically relevant feature as it allows us to statistically test whether the decision maker (Harold Zuercher in the bus engine replacement example) has separable time preferences or not.

While our model allows for more general and richer time preferences, we show that it can be solved and estimated similarly as the standard DDC model.  
In the paper, we first show that agents with our nonseparable time-preferences have well-defined optimal dynamic choices. Based on the Bellman equation with our nonseparable time-preferences, we first prove the existence of the value function.
Specifically, we apply lattice theory and Tarski's fixed point theorem
 to show the existence of the smallest and largest fixed points of a mapping in the space of value functions. 
This allows us to use a simulated nested fixed point algorithm to find the smallest and largest fixed points for model estimation; if the two coincide, then the value function is unique. We also provide general sufficient conditions that ensure a contraction mapping and uniqueness of the value function. These conditions reduce to the standard $\beta <1$ condition for separable preferences but we also consider specific parametrizations (CARA and CRRA) of Epstein-Zin preferences.




We apply our framework to the bus engine replacement data in \cite{rust1987optimal}. In his original work, Rust assumes that the manager Harold Zuercher is risk neutral with separable time preferences. At the beginning of each period, Harold makes a dynamic choice for each bus engine by trading off between an immediate lump sum cost of replacing it and higher maintenance costs for keeping it. In addition to allowing Harold to have non-separable time preferences, we generalize the Rust model in several other aspects.
We assume that the manager may be risk-averse and may benefit from operating the bus (e.g., collecting passenger fares proportional to the additional mileage in each month). 

We estimate the model with non-separable and separable time preferences, respectively, under the CARA parameterization. Our estimation results of the non-separable model suggest that the agent is likely to prefer late resolution of uncertainty. Because the model with separable preferences is nested in the non-separable model, comparing the likelihood from the two models is straightforward, and we are able to statistically reject the null hypothesis that Harold has separable time preferences. 
Moreover, our estimates for the agent's risk attitudes and the payoff parameters are significantly different when allowing for non-separable preferences.
We find that, given the time-separable model, the risk preference parameter is overestimated, while the maintenance cost per mileage is underestimated.

To better interpret the model estimates, 
we compute the certainty equivalent given different model specifications. Specifically, 
we consider a counterfactual scenario in which a subsidy program assists the agent in smoothing revenue across periods.
We solve for the monetary payoff that makes the agent (Harold) indifferent when the uncertainty about incremental mileage in each period is eliminated. 
Our general model with non-separable preferences implies that Harold would accept a sure payment of \$120 each month to stay indifferent. 
In contrast, the certainty equivalent under the misspecified model with separable preferences is \$61.6, which is 48.7\% lower than the value implied by the general model.







\paragraph{Related Literature} There are many empirical applications of DDC models. See \citet{miller1984job} and \citet{buchholz2021semiparametric} for dynamic choices of workers, \citet{pakes1986patents} for patenting decisions, \citet{crawford2005uncertainty}, \citet{hendel2006measuring}, and \citet{gowrisankaran2012dynamics} for dynamic consumer demand. 
As far as we know, there is no existing work in the DDC literature that has incorporated non-separable time preferences.
Our main methodological advance is in developing a feasible specification of DDC models with EZ preferences, arguably the most popular model of non-separable time preferences.
Following \cite{rust1987optimal}, we propose a nested fixed point algorithm for estimating our model. 
Our proposed algorithm can be applied to estimate models with separable and non-separable preferences, and it therefore provides a useful tool for testing various model specifications. 

The use of non-separable models has been common in the macrofinance literature since \citet{kreps1978temporal} and \citet{epsein1989sub}.
For example, \cite{bansalyaron2004} provide a unified explanation for several long-standing puzzles in asset markets using a recursive utility model. \citet{epstein2014much} explore the quantitative implications of recursive utility on temporal resolution of uncertainty. 
We find that the risk preference parameter is significantly overestimated under the misspecified model with separable preferences, suggesting the importance of separating preferences for risk and intertemporal substitution.

In the theoretical literature, \cite{frick2019dynamic} provide an axiomatic analysis of dynamic random utility, focusing on the history-dependency of choices and the type of randomness inherent in these models. \cite{lu2020_repeated} provide  a theoretical model of dynamic stochastic choice in which the agent has non-separable preferences. They provide axiomatic foundations and also consider EZ preferences in an application focusing on its choice-theoretic implications. In contrast, our paper focuses on the DDC model as the workhorse model for analyzing dynamic decision processes in structural econometrics.



The rest of the paper is organized as follows. We describe EZ preferences and the dynamic discrete choice model setup in Section \ref{sec:model}. Theoretical results regarding the existence and uniqueness of the value function are provided in Section \ref{sec:theory}. We discuss the estimation strategy and our empirical application in Section \ref{sec:estimation} and \ref{sec:application}, respectively. Section \ref{sec:conclusion} concludes. The proofs are delegated to the appendices.

\section{Non-separable Time Preferences\label{sec:model}}

\subsection{Setup \label{sec:ezp}}

In this section, we introduce non-separable time preferences in an
infinite-horizon recursive framework. First, recall that under standard
risk and intertemporal preferences, the utility function is given
by 
\begin{equation}
\mathbb{E}_{c}\left[\left(1-\beta\right)u\left(c\right)+\beta\mathbb{E}_{v|c}\left[v\right]\right].\label{eq:risk_sep}
\end{equation}
Here, $u$ is a strictly increasing utility function over a set $\mathcal{C}$
of real-valued payoffs, $\beta\in\left[0,1\right]$ is the discount factor and $v$ is the future continuation value.\footnote{Note that since the value function is recursive, this is equivalent
to $u\left(c\right)+\beta v$ as the $1-\beta$ term is a scalar that
can be factored out.} 
Let $\mathcal{U}$ denote the range of $u$. It is easy to see here
that risk is additively-separable across time.\footnote{To see this point, notice that because of the linearity of the model, (\ref{eq:risk_sep}) reduces to $(1-\beta )\mathbb{E}_c[u(c)]+ \beta \mathbb{E}_v[v]$, where the first expectation captures the risk over current consumption; while the second expectation captures the risk over future consumption.}\label{footnote:add_sep}

Following \citet{kreps1978temporal} and \citet{epsein1989sub}, we
consider a generalization of standard preferences that allow for \textit{non-separable
time preferences}, i.e., where the risk is not additively-separable
across time.\footnote{Due to the non-linearity of $\phi$, the utility function (\ref{eq:ez_preference}) cannot be expressed as the additive form of two expectations, unlike the time-separable model described in footnote \ref{footnote:add_sep}.}

We consider the following representation: 
\begin{equation}
\mathbb{E}_{c}\left[\phi\left(\left(1-\beta\right)u\left(c\right)+\beta\phi^{-1}\left(\mathbb{E}_{v|c}\left[v\right]\right)\right)\right],\label{eq:ez_preference}
\end{equation}
where $\phi$ is a strictly increasing aggregator
function that is either concave or convex.\footnote{ \citet{epsein1989sub} consider
an even more general formulation.} 
Refer to Section 2.2 for parametric examples of the utility function (\ref{eq:ez_preference}), which will aid in understanding the general formulation presented in (\ref{eq:ez_preference}).

As we will see below, the curvature of $\phi$ captures the agent's
attitudes towards how risk is resolved across different time periods.   When $\phi$ is linear, this reduces to the standard time-separable
preferences, which  implies the the indifference to the timing of resolution of uncertainty, otherwise, risk cannot be additively separated
across time as in Equation (\ref{eq:risk_sep}). Note that although
the above formula is recursive and assumes infinite periods, the functional
form can be easily applied to finite time periods by calculating utilities
from the last period via backwards induction. 

The essential difference from separable preferences is the role of
the aggregator $\phi$, which allows for greater flexibility in modeling
risk and time preferences. In fact, agents may exhibit non-trivial
preferences over \textit{when} risk is resolved. 
Consider the two lotteries shown in Figure \ref{fig:earlylate}. Both
lotteries are the same in terms of probabilities and outcomes: they
differ only in the timing of when uncertainty is resolved. In the
left lottery $P$, the uncertainty about period 2 outcome ($10$
or $0$) is resolved in period 1. In the right lottery $Q$, on the
other hand, the uncertainty is resolved in period 2. If the agent always prefers a lottery where uncertainty is resolved in period 1 (e.g., $P$) to a lottery where uncertainty is resolved in period 2 (e.g., $Q$), then we say  that he \textit{prefers
early resolution of uncertainty}. If the agent always prefers a lottery where uncertainty is resolved in period 2 (e.g. $Q$) to a lottery where uncertainty is resolved in period 1 (e.g. $P$), then we say he \textit{prefers late resolution of uncertainty}.
If he is indifferent between the two lotteries, then we say that he is \textit{indifferent to the timing of resolution of uncertainty}.

\begin{figure}[htbp!]
\centering 
\caption{Early vs. late resolution of uncertainty}
\includegraphics[scale=0.35]{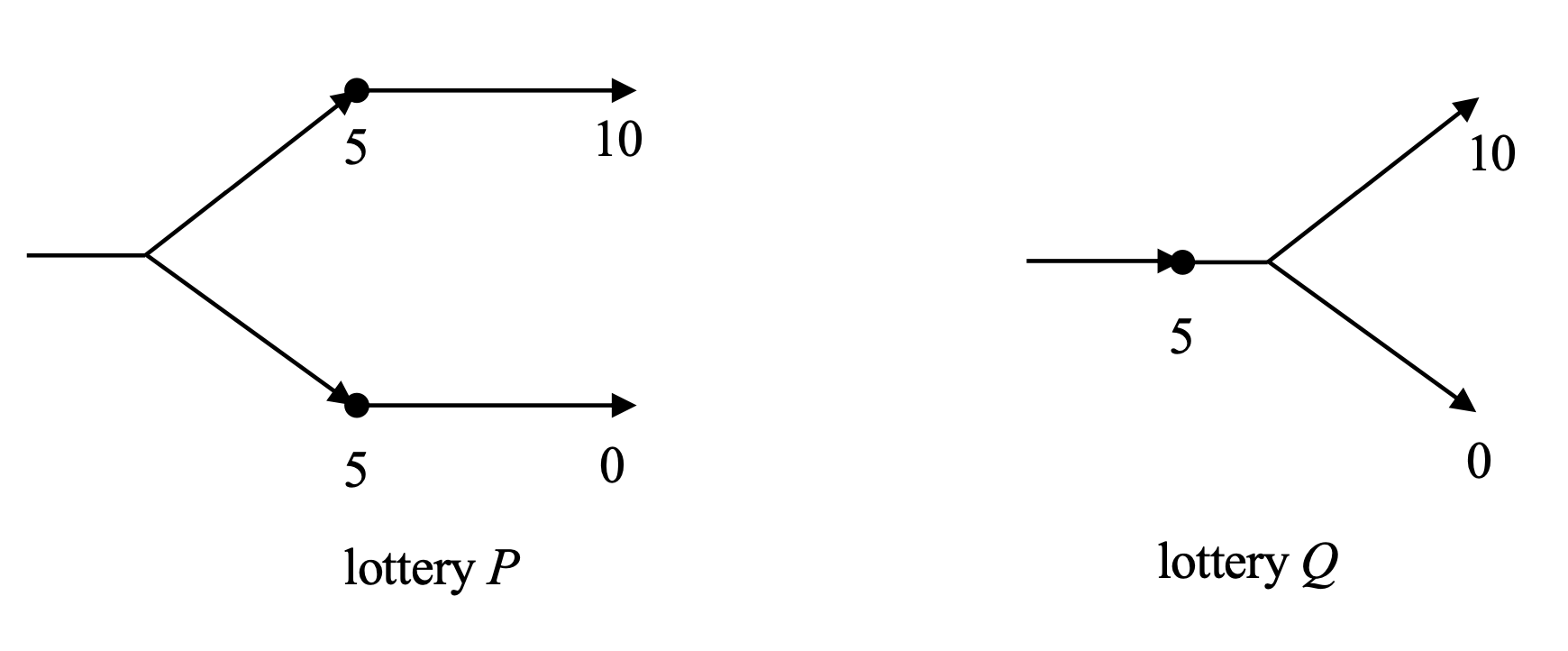} 
\label{fig:earlylate} 
\end{figure}

In our model, preferences for the timing of resolution of uncertainty
are fully captured by the the curvature of of $\phi$. To see why,
consider the example above. Applying \eqref{eq:ez_preference} to
two-periods, the value of lottery $P$ is given by
\[
V\left(P\right)=\mathbb{E}\left[\phi\left(\left(1-\beta\right)u\left(5\right)+\beta z\right)\right]
\]
where $z$ is a random variable that takes on either $\left(1-\beta\right)u\left(0\right)$
or $\left(1-\beta\right)u\left(10\right)$. Note here that all uncertainty
about $z$ is resolved in the first period. On the other hand, the
value of lottery $Q$ is given by
\[
V\left(Q\right)=\phi\left(\left(1-\beta\right)u\left(5\right)+\beta\phi^{-1}\left(\mathbb{E}\left[\phi\left(z\right)\right]\right)\right)
\]
Here, the uncertainty about $z$ is not resolved until the second
period, so the expectation is taken inside the function $\phi^{-1}\left(\cdot\right)$.
If we define $\tilde{\phi}\left(z\right):=\phi\left(\left(1-\beta\right)u\left(5\right)+\beta z\right)$,
then the agent prefers early resolution of uncertainty if $V\left(P\right)\geq V\left(Q\right)$
or
\begin{alignat*}{1}
\mathbb{E}\left[\tilde{\phi}\left(z\right)\right] & \geq\tilde{\phi}\left(\phi^{-1}\left(\mathbb{E}\left[\phi\left(z\right)\right]\right)\right)\\
\tilde{\phi}^{-1}\left(\mathbb{E}\left[\tilde{\phi}\left(z\right)\right]\right) & \geq\phi^{-1}\left(\mathbb{E}\left[\phi\left(z\right)\right]\right)
\end{alignat*}
which holds if $\tilde{\phi}$ is less concave than $\phi$.

We can now generalize how the curvature of $\phi$ captures preference
for early or late resolution of uncertainty. Recall that the Arrow-Pratt
measure of a function $\phi$ is given by 
\[
A_{\phi}\left(z\right):=-\frac{\phi^{\prime\prime}\left(z\right)}{\phi^{\prime}\left(z\right)}.
\]

\begin{proposition}\label{prop:concave} The agent prefers early
(late) resolution of uncertainty iff for all $y\in\mathcal{U}$, $A_{\phi_{y}}\left(z\right)\leq A_{\phi}\left(z\right)$
(resp. $A_{\phi_{y}}\left(z\right)\geq A_{\phi}\left(z\right)$)
where $\phi_{y}\left(z\right):=\phi\left(\left(1-\beta\right)y+\beta z\right)$.
\end{proposition}
\begin{proof}
See \citet{tomasz2013} and \citet{stanca2023}.\footnote{In \citet{tomasz2013}, this result is presented in the context
of ambiguity.}
\end{proof}

To see the condition in Proposition \ref{prop:concave} more explicitly,
note that
\[
A_{\phi_{y}}\left(z\right)=\beta A_{\phi}\left(\left(1-\beta\right)y+\beta z\right)
\]
so the agent prefers early resolution of uncertainty iff
\[
\beta A_{\phi}\left(\left(1-\beta\right)y+\beta z\right)\leq A_{\phi}\left(z\right)
\]
for all $y\in\mathcal{U}$. The condition for preference for late
resolution is the same but with the inequality reversed.

Preferences over the timing of resolution of uncertainty can play
an important role in agents' decision making process in dynamic environments.
Different choices are not only associated with different levels of
uncertainty, but are also related to when that uncertainty is realized
over time. 
For example, in Rust's bus engine replacement problem, replacing
an engine entails current (sure) cost but makes the future uncertainty
of engine failure resolved \emph{earlier}. On the other hand, not
replacing the engine leads to the future uncertainty of engine failure
resolved \emph{later}.


There is no conclusive experimental evidence on whether individuals have preference for early or late resolution of uncertainty. In fact,
agents' preference about the timing of resolution of uncertainty could
be highly context-dependent. For example, when uncertainty is over
positive consumption levels, a preference for early resolution may
be more reasonable. Indeed, in the macro-finance literature, the assumption
of preference for early resolution of uncertainty is often imposed.
On the other hand, in cases when payoffs can be very negative, individuals
may prefer to avoid early resolution of uncertainty. As a heuristic
example, consider a pronouncement of a fatal disease from a doctor.
Some patients may prefer to learn the news later because this information
can induce a great deal of disappointment, fear and anxiety. These
negative emotions, while may not have instrumental values, could still
directly affect their utility.\footnote{Note that in the definition of a preference for the timing of resolution
of uncertainty, the information must be non instrumental, i.e., the
agent cannot do anything given the information.}



\subsection{Parametric Special Cases\label{sec:para}}

We now present several parametrizations of non-separable time preferences
that have been used in the literature.

\subsubsection{CRRA Epstein-Zin}

By far the most common parametric form of Epstein-Zin preferences
is when both $u$ and $\phi$ have CRRA (or power) functional forms.
Here, the utility is given by $u\left(c\right)=c^{1-\rho}$ and the
time aggregator is given by 
\[
\phi\left(z\right)=\left(z^{\frac{1}{1-\rho}}\right)^{1-\alpha}=z^{\frac{1-\alpha}{1-\rho}}
\]
where $\alpha\leq1$ and $\rho\leq1$ are parameters that characterize intertemporal and
risk preferences respectively. Under these
assumptions, the value function (\ref{eq:ez_preference}) yields the
following form: 
\begin{alignat}{1}
\mathbb{E}_{c}\left[\left(\left(1-\beta\right)c^{1-\rho}+\beta\mathbb{E}_{v|c}\left[v\right]^{\frac{1-\rho}{1-\alpha}}\right)^{\frac{1-\alpha}{1-\rho}}\right].\label{eq:crra_ez}
\end{alignat}

To see how  $\rho$ and $\alpha$ characterize intertemporal and
risk preferences, first consider the case where there is no risk so we can remove the expectations in \eqref{eq:ez_preference}. Redefining a new value function as $\tilde{v}=\phi^{-1}\left(v\right)$, we obtain from \eqref{eq:ez_preference} that the new value function is given by
\[
\left(1-\beta\right)u\left(c\right)+\beta\tilde{v}=\left(1-\beta\right)c^{1-\rho}+\beta \tilde{v}
\]
where intertemporal preferences are determined solely by $\rho$. Here,
$\rho$ measures the elasticity of substitution, defined as the percentage
change of the growth rate of consumption (i.e. $c_{t+1}/c_{t}$) to
the percentage change of the intertemporal marginal rate of substitution
(MRS):\footnote{Recall that MRS is defined as $u^{\prime}\left(c_{t+1}\right)/u^{\prime}\left(c_{t}\right)$.}
\begin{equation}
\dfrac{d\log\left(c_{t+1}/c_{t}\right)}{d\log MRS_{t,t+1}}=\rho^{-1}.\label{eq:eis2}
\end{equation}
For risk preferences, consider a one-time shock that resolves within
a single-period followed by zero consumption forever. It straightforward
to see that the Bernoulli utility in \eqref{eq:crra_ez} reduces to
$c^{1-\alpha}$, which has standard constant relative risk aversion
(CRRA) parameter $\alpha$. The fact that both intertemporal and risk
preferences can be so neatly separated here (which is not possible
with standard separable preferences) speaks to its popularity in the
literature.

Proposition \ref{prop:concave} reduces to the following in this parametrization.

\begin{corollary} \label{corr:earlyCRRA}An agent with CRRA Epstein-Zin preferences prefers
early (late) resolution of uncertainty iff $\rho\leq\alpha$ (resp.
$\rho\geq\alpha$) iff $\phi$ is concave (resp. convex). \end{corollary} 
\begin{proof}
Since $\phi\left(z\right)=z^{\frac{1-\alpha}{1-\rho}}$, 
\begin{alignat*}{1}
A_{\phi}\left(z\right) & =\left(1-\frac{1-\alpha}{1-\rho}\right)z^{-1},\\
A_{\phi_{y}}\left(z\right) & =\beta\left(1-\frac{1-\alpha}{1-\rho}\right)\left(\left(1-\beta\right)y+\beta z\right)^{-1}.
\end{alignat*}
Thus, $A_{\phi_{y}}\left(z\right)\leq A_{\phi}\left(z\right)$
iff $\frac{1-\alpha}{1-\rho}\leq1$ or $\rho\leq\alpha$. The result
follows from Proposition \ref{prop:concave} and noting that $\phi$
is concave (convex) iff $\rho\leq\alpha$ (resp. $\rho\geq\alpha$). 
\end{proof}

\subsubsection{CARA Epstein-Zin}
\label{sec:cara}

An alternative parameterization is when $u$ has the CARA (or exponential)
functional form. Here, the utility is given by $u\left(c\right)=u_{\rho}\left(c\right)$
where
\begin{equation}
u_{\rho}\left(c\right):=\rho^{-1}\left(1-e^{-\rho c}\right)\label{eq:cara}
\end{equation}
and the time aggregator is given by
\[
\phi\left(z\right)=u_{\alpha}\left(u_{\rho}^{-1}\left(z\right)\right)=\alpha^{-1}\left(1-\left(1-\rho z\right)^{\frac{\alpha}{\rho}}\right)
\]
for parameters $\alpha\geq0$ and $\rho\geq0$ that characterize intertemporal and
risk preferences. Under these assumptions,
the value function (\ref{eq:ez_preference}) yields the following
form
\begin{alignat}{1}
\mathbb{E}_{c}\left[\alpha^{-1}\left(1-\left(\left(1-\beta\right)e^{-\rho c}+\beta\left(1-\mathbb{E}_{v|c}\left[\alpha v\right]\right){}^{\frac{\rho}{\alpha}}\right){}^{\frac{\alpha}{\rho}}\right)\right].\label{eq:cara_ez}
\end{alignat}
These preferences have appeared in \citet{skiadas2009asset} and are
exactly the CARA analog of CRRA Epstein-Zin preferences. One advantage
of this formulation is that it allows for negative payoffs, which
CRRA cannot accommodate. For example, in \citet{rust1987optimal},
the agent needs to pay a large cost when choosing to replace the bus
engine resulting in net negative revenue that period.%

As in CRRA Epstein-Zin, the parameters $\rho$ and $\alpha$ characterize
intertemporal and risk preferences. First, in the absence of risk, we can follow the same reasoning as for CRRA Epstein-Zin above by redefining a new value function $\tilde{v}=\phi^{-1}\left(v\right)$. We now obtain from \eqref{eq:ez_preference} that the new value function is given by
\[
\left(1-\beta\right)u\left(c\right)+\beta\tilde{v}=\left(1-\beta\right)\rho^{-1}\left(1-e^{-\rho c}\right)+\beta \tilde{v}
\]
so intertemporal preferences are determined solely by $\rho$. Here,
in contrast to CRRA Epstein-Zin, $\rho$ measures the responsiveness
of the \textit{difference} in consumption between two periods (i.e.
$c_{t+1}-c_{t}$) to the percentage change of MRS:
\begin{equation}
\dfrac{d\left(c_{t+1}-c_{t}\right)}{d\log MRS_{t,t+1}}=\rho^{-1}.\label{eq:eis}
\end{equation}
For risk preferences, as before, consider a one-time shock that resolves
within a single-period followed by zero consumption forever. It straightforward
to see that the Bernoulli utility in \eqref{eq:cara_ez} has Arrow-Pratt
coefficient 
\[
\frac{\left(1-\beta\right)\alpha+\beta e^{\rho c}\rho}{\left(1-\beta\right)+\beta e^{\rho c}}.
\]
This is a convex combination between $\alpha$ and $\rho$ where the
weight on $\alpha$ is decreasing in $\beta$ and consumption $c$.
Note that holding all else constant, risk aversion is increasing in
$\alpha$. Unlike CRRA Epstein-Zin preferences, while intertemporal
preferences depend solely on $\rho$, risk preferences here depend
on both $\alpha$ and $\rho$. 

Proposition \ref{prop:concave} reduces to the following in this parametrization.

\begin{corollary}An agent with CARA Epstein-Zin preferences prefers
early (late) resolution of uncertainty iff $\rho\leq\alpha$ (resp.
$\rho\geq\alpha$) iff $\phi$ is concave (resp. convex). \end{corollary} 
\begin{proof}
Since $\phi\left(z\right)=\alpha^{-1}\left(1-\left(1-\rho z\right)^{\frac{\alpha}{\rho}}\right)$,
\begin{alignat*}{1}
A_{\phi}\left(z\right) & =\frac{\alpha-\rho}{1-\rho z},\\
A_{\phi_{y}}\left(z\right) & =\frac{\beta\left(\alpha-\rho\right)}{1-\rho\left(\left(1-\beta\right)y+\beta z\right)}.
\end{alignat*}
Thus, $A_{\phi_{y}}\left(z\right)\leq A_{\phi}\left(z\right)$
iff $\rho\leq\alpha$ (noting that $y\leq\rho^{-1}$ given the range
of $u=u_{\rho}$). The result follows from Proposition \ref{prop:concave}
and noting that $\phi$ is concave (convex) iff $\rho\leq\alpha$
(resp. $\rho\geq\alpha$).
\end{proof}

\subsubsection{Separable Preferences}

Finally, both CRRA and CARA Epstein-Zin preferences nest separable
preferences as special cases. When $\alpha=\rho$ in either case,
the time aggregator becomes linear, i.e. $\phi\left(z\right)=z$,
and the value function becomes
\begin{equation}
\left(1-\beta\right)\mathbb{E}_{c}\left[u\left(c\right)\right]+\beta\mathbb{E}_{v}\left[v\right]\label{eq:sep}
\end{equation}
where $u$ is either of the CRRA or CARA functional form. Furthermore,
when $\alpha=\rho\rightarrow0$ under either, $u\left(c\right)=c$
so this reduces to the separable risk-neutral case. This is the specification
commonly adopted in many empirical studies, including the engine replacement
model in \citet{rust1987optimal}. 

The fact that separable preferences, i.e. $\alpha=\rho$, is nested
in our more general parametrizations allows us to statistically test
whether the agent has separable preferences. While separability provides
computational convenience, misspecifying the agent's preference as
additively time separable when it is not can lead to systematic biases
in inferences of the structural primitives of the model. We see this
for the case of CARA Epstein-Zin preferences in Section \ref{sec:application}.

\section{Dynamic Discrete Choice: General Setup and Theoretical Results}\label{sec:theory}


We now apply the general non-separable preferences defined
in the previous section to a general dynamic discrete choice framework.
We provide two key theoretical results. The first result pertains to the existence of
a value function, while the second addresses its uniqueness.



\subsection{General Setup}
\label{sec:ddc}
Denote the set of states by $\mathcal{X}$ and sets of actions by $\mathcal{D}$. We
assume $\mathcal{D}$ is finite. The timing of our setup is as follows. 
At the
beginning of period $t$, 
the state $x_t$ realizes and then an idiosyncratic shock $\varepsilon_{t}$ realizes. The agent then makes a choice $d_{t}\in \mathcal{D}$.
Finally, the uncertainty about another variable $\Delta_t$, which affects the current period consumption is realized.
The agent consumes the realized consumption that period, which is given by
\begin{equation}
c(d_{t},x_{t},\Delta_t,\varepsilon_{t})=\pi(d_{t},x_{t},\Delta_t)+\varepsilon_{t}(d_{t}),\label{eq:current_payoff}
\end{equation}
where $\pi(d_{t},x_{t},\Delta_{t})$ represents the deterministic consumption.
We assume that $\varepsilon_{t}(d_{t})$ is a random shock for alternative $d_t$. The vector of random shocks $\varepsilon_t = (\varepsilon_t(1), \cdots, \varepsilon_t(\abs{\mathcal{D}}))$ is drawn from a cumulative distribution function $G(\cdot)$.


We represent the agent's optimization problem using the Bellman’s equation as follows.
\begin{align}
    V(x_t, \varepsilon_t) & = \max_{d \in \mathcal{D}} \bigg\{ 
    \mathbb{E}_{\Delta_{t}|d,x_{t}}\bigg[\phi\bigg(\left(1-\beta\right)u\big(c\left(d,x_{t},\Delta_{t},\varepsilon_{t}\right)\big) \nonumber \\
    & \quad \quad \quad  + \beta\phi^{-1}\big(\mathbb{E}_{x_{t+1},\varepsilon_{t+1} | d, x_t, \Delta_t}[V(x_{t+1}, \varepsilon_{t+1})]\big)\bigg)\bigg]
    \bigg\}. \label{eq:bellman}
\end{align}
The first expectation is due to uncertainty in the current-period payoff shock $\Delta_{t}$, and the second expectation is due to uncertainty in the future state $(x_{t+1}, \varepsilon_{t+1})$. We follow the dynamic discrete choice literature to assume that (1) the random shock $\varepsilon$'s are independent over time and (2) conditional on the current state $x_t$ and choice $d_t$, the future state $x_{t+1}$ is independent of the unobserved state $\varepsilon_t$. Given these assumptions, the distribution of future state $(x_{t+1}, \varepsilon_{t+1})$ does not depend on $\varepsilon_t$. We therefore omit $\varepsilon_t$ in the second expectation $\mathbb{E}_{x_{t+1},\varepsilon_{t+1} | d, x_t, \Delta_t}$.

From the Bellman's Equation in (\ref{eq:bellman}), we define the ex-ante value function as
\begin{align}
\label{va2V}
    V(x_{t+1}) = \int V(x_{t+1}, \varepsilon_{t+1}) G(\varepsilon_{t+1}).
\end{align}
Note that $V(x_{t+1})$ represent the agent's expected value at the beginning of period $t+1$ before the realization of the future shock $\varepsilon_{t+1}$. 
If the agent chooses $d_t$, the value she receives can be represented as
\begin{equation}
\begin{split}
&v\left(d_{t},x_{t},\varepsilon_{t}\right)\\
&=\mathbb{E}_{\Delta_{t}|d_{t},x_{t}}\bigg[\phi\bigg(\left(1-\beta\right)u\big(c\left(d_{t},x_{t},\Delta_{t},\varepsilon_{t}\right)\big)+\beta\phi^{-1}\big(\mathbb{E}_{x_{t+1} | d_t, x_t, \Delta_t}[V(x_{t+1})]\big)\bigg)\bigg].\label{V2va}
\end{split}
\end{equation}
Note that $v\left(d_{t},x_{t},\varepsilon_{t}\right)$ is essentially the ``choice-specific'' value function defined in the dynamic discrete choice literature. 
The value functions in Equation \eqref{V2va} are non-separable in random shocks $\varepsilon_{t}$ even if we assume
that those shocks enter the current payoff additively as in Equation (\ref{eq:current_payoff}).
In the special case when $\phi$ is linear and the agent is risk-neutral, $\varepsilon_t$ is additively separable from other expressions in the value function and this reduces to the standard discrete choice setup.

The decision maker's optimal choice is defined as 
\[
d_{t}^{\ast}=\arg\max_{d \in \mathcal{D}}\text{ } v(d,x_{t},\varepsilon_{t}), 
\]
which leads to a Conditional Choice Probability (CCP) of choosing action
$d_t$ 
\begin{equation}
p(d_t|x_{t})=\int\mathbf{1}\bigg\{ d_t \in\arg\max_{d\in \mathcal{D}} \text{ }v(d,x_{t},\varepsilon_{t})\bigg\} dG(\varepsilon_{t}).\label{ccp}
\end{equation}
With the non-separability of random shocks, the CCP in Equation (\ref{ccp}) and the ex-ante
value function in Equation (\ref{va2V}) do not have closed-form solutions even if we assume that $\varepsilon$'s follow extreme value distributions.

Combining Equations \eqref{va2V} and \eqref{V2va}, we can define
a value function iteration operator $T$ in the space of the ex-ante
value functions. 
To simplify notation, define
\[
r(d_{t},x_{t}, \Delta_t,\varepsilon_{t},V):=\phi\left(\left(1-\beta\right)u\left(c\left(d_{t},x_{t}, \Delta_t, \varepsilon_{t}\right)\right)+\beta\phi^{-1}\left( \mathbb{E}_{x_{t+1}| d_t, x_t, \Delta_t}V\left(x_{t+1}\right)\right)\right).
\]
Then the operator $T$ is defined as 
\begin{equation}
T(V)(x_{t})=\int\max_{d_{t}\in \mathcal{D}}\bigg\{\int r(d_{t},x_{t},\Delta_{t},\varepsilon_{t},V)dF_{d_{t}, x_t}(\Delta_{t})\bigg\} dG(\varepsilon_{t}),\label{eq:contraction}
\end{equation}
where $F_{d_{t},x_{t}}(\Delta_t)$ represents the cumulative distribution
function of $\Delta_{t}$ given the current state $x_{t}$ and the agent's
choice $d_{t}$. The ex-ante continuation value $V$ is obtained as a solution  of the functional equation $T(V)=V$.


\begin{remark}
    When $\phi$ is a linear function, this model reduces to a standard dynamic discrete choice model with time-separable preferences and the choice-specific value function becomes (assuming $\phi(x)=x$ for simplicity)
\begin{equation}
v\left(d_{t},x_{t},\varepsilon_{t}\right)=\left(1-\beta\right)\mathbb{E}_{\Delta_t | d_t, x_t}\big[u\big(c\left(d_{t},x_{t},\Delta_t,\varepsilon_{t}\right)\big)\big]+\beta\mathbb{E}_{x_{t+1} | d_t, x_t}\big[V(x_{t+1})\big].\label{V2va_sep}
\end{equation}
Note that in this case, the current-period payoff shocks $\Delta_t$ are completely separated across time. 
\end{remark}

\begin{remark}
    When there is no current-period uncertainty (i.e., $\Delta_t$ does not enter the utility), the choice-specific value function can be simplified to
\begin{equation}
v\left(d_{t},x_{t},\varepsilon_{t}\right)=\phi\bigg(\left(1-\beta\right)u\big(c\left(d_{t},x_{t},\varepsilon_{t}\right)\big)+\beta\phi^{-1}\big(\mathbb{E}_{x_{t+1} | d_t, x_t}[V(x_{t+1})]\big)\bigg).\label{V2va_simp}
\end{equation}
Note that this does not reduce to the standard dynamic discrete choice setup. The reason is that there are uncertainties with respect to the future state $(x_{t+1},\varepsilon_{t+1})$, and expectations with respect to these uncertainties are taken in between the aggregators $\phi^{-1}$ and $\phi$. Whenever preferences are non-separable so $\phi$ is non-linear, Equation (\ref{V2va_simp}) differs from standard dynamic discrete choice model even if there is no current-period uncertainty.
\end{remark}

\subsection{Existence: Value Function}

In this section, we establish existence of a value function. We show that the operator $T$ defined in Equation \eqref{eq:contraction} always has a
fixed point under some mild assumptions. Moreover, we demonstrate how to obtain the largest and
smallest fixed points. 

Let $\underline{\pi}$ and $\bar{\pi}$ denote the smallest and largest
possible value of profit $\pi$. Define 
\begin{alignat*}{1}
v^{\ast} & :=\max\left\{ \mathbb{E}\left[\phi\left(u\left(\bar{\pi}+\hat{\varepsilon}\right)\right)\right],\phi\left(\mathbb{E}\left[u\left(\bar{\pi}+\hat{\varepsilon}\right)\right]\right)\right\} ,\\
v_{\ast} & :=\min\left\{ \mathbb{E}\left[\phi\left(u\left(\underline{\pi}+\varepsilon\right)\right)\right],\phi\left(\mathbb{E}\left[u\left(\underline{\pi}+\varepsilon\right)\right]\right)\right\} ,
\end{alignat*}
where $\hat{\varepsilon}:=\max_{d\in\mathcal{D}}\varepsilon_{d}$. These are bounds for the largest and smallest possible values of continuation values. Suppose $v^{\ast}$ and $v_{\ast}$ are both finite and let $\mathcal{V}$
be the set of value functions bounded between some $\underline{v}\leq v_{\ast}$
and $\bar{v}\geq v^{\ast}$, that is, the set of $V$ such that $\underline{v}\leq V\left(x\right)\leq\bar{v}$
for all $x\in\mathcal{X}$.

We first show that $T:\mathcal{V}\rightarrow\mathcal{V}$ is well-defined
(see Appendix \ref{sec:proof_theorem_1}). We then define a complete
lattice $\left(\mathcal{V},\geq\right)$, where $\geq$ is the partial
order on $\mathcal{V}$ such that $V\geq V^{\prime}$ if $V\left(x\right)\geq V^{\prime}\left(x\right)$
for all $x\in\mathcal{X}$.\footnote{A partially ordered set $\left(\mathcal{V},\geq\right)$ is a complete
lattice if every subset has both a supremum and an infimum in $\mathcal{V}$
according to the partial order $\geq$.} It is easy to see that $T$ is monotonic, that is, $V\geq V^{\prime}$
implies $T\left(V\right)\geq T\left(V^{\prime}\right)$.

\begin{lemma}\label{lem:monotone} $T$ is monotonic. \end{lemma} 
\begin{proof}
Follows from the fact that $u$ and $\phi$ are all increasing functions
so $r$ is also increasing in $V$. 
\end{proof}
Since $T$ is monotonic, Tarski's fixed point theorem (see \citet{AB2006infinite})
guarantees that the set of fixed points of $T$ is non-empty and a
complete lattice. Moreover, if we start iterating from the smallest
and largest possible values of $V$, then we obtain the respective
smallest and largest fixed points of $T$. We summarize this as follows.

\begin{theorem} \label{theorem:existence} Suppose $v^{\ast}$ and
$v_{\ast}$ are finite and let $\mathcal{V}$ be the set of functions
bounded by some $\underline{v}\leq v_{\ast}$ and $\bar{v}\geq v^{\ast}$.
Then $T:\mathcal{V}\rightarrow\mathcal{V}$ has a fixed point. Moreover,
$\lim_{n}T^{n}\left(\underline{v}\right)$ and $\lim_{n}T^{n}\left(\bar{v}\right)$
are its smallest and largest fixed points respectively. \end{theorem} 
\begin{proof}
See Appendix \ref{sec:proof_theorem_1}. 
\end{proof}

Existing results assume separable time preferences. Our contribution is demonstrating that for non-separable time preferences, we can apply Tarski's theorem as long as the $T$ mapping
is monotonic (i.e. increasing) in the value function.\footnote{ \citet{jia2008} also uses Tarski's fixed point theorem but for a different problem (entry game). She also assumes separable time preferences.} This also provides
a way to compute the smallest and largest value functions by iterating
$T$. Although this does not guarantee a unique fixed point, we can
easily check by calculating the smallest and largest fixed points
and seeing if they coincide.

How about specific parametrizations of non-separable preferences?
Theorem \ref{theorem:existence} implies the following for CRRA Epstein-Zin
preferences.

\begin{corollary}Suppose the agent has CRRA Epstein-Zin preferences.
If $\mathbb{E}\left[\varepsilon\right]$ is finite, then $T$
has a fixed point.\end{corollary}
\begin{proof}
Since both $u$ and $\phi$ are both power functions, they both yield
positive values so $v_{\ast}$ must be finite. Next, note that 
\begin{alignat*}{1}
\mathbb{E}\left[u\left(\bar{\pi}+\hat{\varepsilon}\right)\right] & =\mathbb{E}\left[\left(\bar{\pi}+\hat{\varepsilon}\right)^{1-\rho}\right]\leq\left(\bar{\pi}+\mathbb{E}\left[\max_{d\in\mathcal{D}}\varepsilon_{d}\right]\right)^{1-\rho}
\end{alignat*}
where the inequality follows from the fact that $\rho\leq1$. Since
$\max_{d\in\mathcal{D}}\varepsilon_{d}\leq\sum_{d\in\in\mathcal{D}}\varepsilon_{d}$,
it has finite expectation. The reasoning for $\mathbb{E}\left[\phi\left(u\left(\bar{\pi}+\hat{\varepsilon}\right)\right)\right]=\mathbb{E}\left[\left(\bar{\pi}+\hat{\varepsilon}\right)^{1-\alpha}\right]$
is symmetric, so $v^{\ast}$ is finite and the result follows from
Theorem \ref{theorem:existence}. 
\end{proof}
With CARA Epstein-Zin preferences, Theorem \ref{theorem:existence}
implies the following.

\begin{corollary}\label{corr:existence}Suppose the agent has CARA
Epstein-Zin preferences. If $\mathbb{E}\left[e^{-t\varepsilon}\right]$
is finite for all $t\in\mathbb{R}$, then $T$ has
a fixed point. \end{corollary}
\begin{proof}
Since $\phi=u_{\alpha}\circ u_{\rho}^{-1}$ where $u_{\alpha}$ is
CARA, $\phi$ is bounded above by $\alpha^{-1}$. Thus, $v^{\ast}=\alpha^{-1}$
is finite. Next, note that 
\begin{alignat*}{1}
\mathbb{E}\left[u\left(\underline{\pi}+\varepsilon\right)\right] & =\mathbb{E}\left[u_{\rho}\left(\underline{\pi}+\varepsilon\right)\right]=\rho^{-1}\left(1-e^{-\rho\underline{\pi}}\mathbb{E}\left[e^{-\rho\varepsilon}\right]\right)
\end{alignat*}
which is finite given the assumption. The reasoning for $\mathbb{E}\left[\phi\left(u\left(\underline{\pi}+\varepsilon\right)\right)\right]=\mathbb{E}\left[u_{\alpha}\left(\underline{\pi}+\varepsilon\right)\right]$
is symmetric, so $v_{\ast}$ is finite and the result follows from
Theorem \ref{theorem:existence}.
\end{proof}
The condition in Corollary \ref{corr:existence} is weak; it is equivalent
to saying that the moment generating function for $-\varepsilon$
always exists. This is true if $\varepsilon$ is Normal or Extreme-Value
distributed, which we note below.

\begin{corollary} Suppose the agent has CARA Epstein-Zin preferences
and $\varepsilon$ either has a Normal or Extreme Value Type I distribution.
Then starting with the initial point $v^{\ast}=\alpha^{-1}$, $\lim_{n}T^{n}\left(v^{\ast}\right)$
converges to its largest fixed point. \end{corollary}
\begin{proof}
Note that both the Normal distribution and the Extreme Value Type
I (Gumbel) distribution have well-defined (finite) moment generating
functions. It is easy to see that the same holds for $-\varepsilon$.
\end{proof}

\subsection{Uniqueness: Contraction Mapping}

In this section, we present conditions under which $T$ constitutes
a contraction, ensuring the uniqueness of the value function. First,
recall from Proposition \ref{prop:concave} that $\phi_{y}\left(z\right)=\phi\left(\left(1-\beta\right)y+\beta z\right)$
and define
\[
\psi_{y}\left(z\right):=\phi_{y}\left(\phi^{-1}\left(z\right)\right).
\]
Recall that $\underline{\pi}$ and $\bar{\pi}$ are the smallest and
largest possible value of profit $\pi$. Also recall that $\underline{v}$
and $\bar{v}$ are the bounds of the value function (which can be
infinite). We now present our second theorem.

\begin{theorem}\label{theorem:contraction}$T$ is a contraction
mapping if
\begin{equation}
M:=\mathbb{E}\left[\max_{d\in D}\sup_{\pi\in\left[\underline{\pi},\bar{\pi}\right],z\in\left[\underline{v},\bar{v}\right]}\psi_{u\left(\pi+\varepsilon_{d}\right)}^{\prime}\left(z\right)\right]<1\label{eq:M}
\end{equation}
\end{theorem} 
\begin{proof}
See Appendix \ref{sec:proof_theorem_2}. 
\end{proof}
Theorem \ref{theorem:contraction} implies that when the bound $M<1$,
the ex-ante value function $V(\cdot)$ is a unique fixed point of
the contraction mapping $T$ as defined in equation (\ref{eq:contraction}).
For intuition, note that $\psi_{y}\left(z\right)$ captures
the value given current period utility $y$ and future continuation
value $z$. For instance, if $\phi\left(z\right)=z$ as in the standard
separable case, then
\[
\psi_{y}\left(z\right)=\left(1-\beta\right)y+\beta z.
\]
In this case, $\psi_{y}\left(\cdot\right)$ is linear so $\psi_{y}^{\prime}\left(z\right)=\beta$.
This immediately yields the standard result that coincides with Rust's
condition for a contraction mapping.

\begin{corollary}Suppose the agent has separable preferences. If
$\beta<1$, then $T$ is a contraction mapping.\end{corollary}


In general, when the agent does not have separable preferences, the
bounds on $\psi_{y}^{\prime}\left(z\right)$ would depend on the curvature
of $\phi$. More explicitly, we have
\begin{alignat*}{1}
\psi_{y}^{\prime}\left(z\right) & =\frac{\phi_{y}^{\prime}\left(\tilde{z}\right)}{\phi^{\prime}\left(\tilde{z}\right)}
\end{alignat*}
where $\tilde{z}=\phi^{-1}\left(z\right)$. We thus have
\[
\psi_{y}^{\prime\prime}\left(z\right)=\frac{\phi_{y}^{\prime}\left(\tilde{z}\right)}{\phi^{\prime}\left(\tilde{z}\right)^{2}}\left(A_{\phi}\left(\tilde{z}\right)-A_{\phi_{y}}\left(\tilde{z}\right)\right).
\]
From Proposition \ref{prop:concave}, this implies that $\psi_{y}^{\prime}\left(\cdot\right)$
is increasing (decreasing) if the agent prefers early (resp. late)
resolution of uncertainty. This implies the following.

\begin{corollary}Suppose the agent prefers early (late) resolution
of uncertainty. If $\mathbb{E}\left[\sup_{d\in D,\pi\in\left[\underline{\pi},\bar{\pi}\right]}\psi_{u\left(\pi+\varepsilon_{d}\right)}^{\prime}\left(\hat{v}\right)\right]<1$
where $\hat{v}=\bar{v}$ (resp. $\hat{v}=\underline{v}$), then $T$
is a contraction mapping. \end{corollary}
\begin{proof}
Follows from Proposition \ref{prop:concave} and Theorem \ref{theorem:contraction}.
\end{proof}
Note that this bound can be further simplified if we know whether
$\phi$ is convex or concave. For instance, if $\phi$ is convex,
then $\psi_{y}^{\prime}\left(z\right)$ is increasing in $y$ so we
can use $\pi=\bar{\pi}$ as the bound when evaluating $M$ in Equation
(\ref{eq:M}). The case for concave $\phi$ is symmetric. 

How about specific parametrizations of non-separable preferences?
Theorem \ref{theorem:contraction} implies the following for CRRA
Epstein Zin preferences.

\begin{corollary}\label{coro:contraction_CRRA}Suppose the agent
has CRRA Epstein-Zin preferences. If $\rho\leq\alpha$ and $\beta^{\frac{1-\alpha}{1-\rho}}<1$,
then $T$ is a contraction mapping. \end{corollary}
\begin{proof}
Since $\phi\left(z\right)=z^{\frac{1-\alpha}{1-\rho}}$, 
\begin{alignat*}{1}
\psi_{y}^{\prime}\left(z\right) & =\frac{\beta\left(\left(1-\beta\right)y+\beta z^{\frac{1-\rho}{1-\alpha}}\right)^{\frac{1-\alpha}{1-\rho}-1}}{\left(z^{\frac{1-\rho}{1-\alpha}}\right)^{\frac{1-\alpha}{1-\rho}-1}}\\
 & =\beta\left(\left(1-\beta\right)yz^{\frac{1-\alpha}{1-\rho}}+\beta\right)^{\frac{1-\alpha}{1-\rho}-1}.
\end{alignat*}
If $\rho\leq\alpha$, then $\frac{1-\alpha}{1-\rho}-1\leq0$. Since
$y\geq0$, this means
\[
\psi_{y}^{\prime}\left(z\right)\leq\beta\left(\beta\right)^{\frac{1-\alpha}{1-\rho}-1}=\beta^{\frac{1-\alpha}{1-\rho}}.
\]
Thus $\beta^{\frac{1-\alpha}{1-\rho}}<1$ ensures a contraction mapping.
\end{proof}
Note that in the other $\rho>\alpha$ case, we can still use Theorem
\ref{theorem:contraction} to obtain a sufficient condition for a
contraction mapping but the bound $M$ in Equation (\ref{eq:M}) would
depend on the error distribution. The reason we are able to obtain a simple sufficient condition for when the agent prefers early resolution of uncertainty ($\rho\leq\alpha$) is that the value function under CRRA Epstein-Zin is bounded below (since $u(c)$ has to be non-negative). The fact that values are bounded below and $\phi$ is concave given  $\rho\leq\alpha$ (see Corollary \ref{corr:earlyCRRA}), allows us to create bounds to ensure a contraction mapping.


With CARA Epstein-Zin preferences, Theorem \ref{theorem:contraction}
implies the following.

\begin{corollary}\label{coro:contraction}Suppose the agent has CARA
Epstein-Zin preferences. If $\rho\geq\alpha$ and $\beta^{\frac{\alpha}{\rho}}<1$,
then $T$ is a contraction mapping. \end{corollary}
\begin{proof}
Since $\phi=u_{\alpha}\circ u_{\rho}^{-1}$, 
\begin{alignat*}{1}
\psi_{y}^{\prime}\left(z\right) & =\frac{\beta\left(\left(1-\beta\right)\left(1-\rho y\right)+\beta\left(1-\alpha z\right)^{\frac{\rho}{\alpha}}\right)^{\frac{\alpha}{\rho}-1}}{\left(\left(1-\alpha z\right)^{\frac{\rho}{\alpha}}\right)^{\frac{\alpha}{\rho}-1}}\\
 & =\beta\left(\left(1-\beta\right)\frac{1-\rho y}{\left(1-\alpha z\right)^{\frac{\rho}{\alpha}}}+\beta\right)^{\frac{\alpha}{\rho}-1}.
\end{alignat*}
Since $\rho\geq\alpha$, $\frac{\alpha}{\rho}-1\leq0$ so
\[
\psi_{y}^{\prime}\left(z\right)\leq\beta\left(\beta^{\frac{\alpha}{\rho}-1}\right)=\beta^{\frac{\alpha}{\rho}}.
\]
Thus $\beta^{\frac{\alpha}{\rho}}<1$ ensures a contraction mapping. 
\end{proof}
In other words, for an agent with CARA Epstein-Zin preferences and
prefers late resolution of uncertainty (i.e. $\rho\geq\alpha$), $T$
is a contraction mapping as long as $\beta$ is sufficiently small
(i.e. $\beta^{\frac{\alpha}{\rho}}<1$). The fixed point would be
unique. As with the CRRA Epstein-Zin case, we can still use Theorem
\ref{theorem:contraction} for the case when $\rho<\alpha$ but the
bound $M$ would again depend on the error distribution. Note that in contrast to CRRA Epstein-Zin, values are bounded above under CARA Epstein-Zin, so it is when the agent prefers late resolution of uncertainty ($\rho\geq\alpha$) where we can obtain simple sufficient conditions for a contraction mapping.

\section{Estimation}\label{sec:estimation}


We propose to use an estimation approach similar to \cite{rust1987optimal}'s Nested Fixed Point algorithm to estimate the structural parameters in the dynamic model described in Section \ref{sec:ddc}.\footnote{Since \cite{rust1987optimal}, several estimation methods for standard DDC models became available, such as the CCP approach of \citet{hotz1993conditional}, the nested pseudo-likelihood approach of \citet{aguirregabiria2002swapping} and their extensions to various settings. See \citet{aguirregabiria2010dynamic} for a survey of methodologies. Recently, \citet{igami2017estimating} and \citet{igami2020mergers} extend the Nested Fixed Point algorithm to dynamic oligopoly games with sequential or stochastically alternating moves.}
In our model, the choice-specific value function is not additively separable in the idiosyncratic shock $\varepsilon$ due to the nonlinear nature of the non-separable time preferences. Therefore, the well-known multinomial logit formula for the conditional choice probabilities and the McFadden's social surplus function are no longer applicable. We describe in this section how we obtain the conditional choice probabilities using simulation methods. Our nested fixed point algorithm contains an inner loop that solves a DDC problem for a given set of parameter values and an outer loop that maximizes the log-likelihood.    

We focus on the estimation of two parameterizations of the model introduced in Section \ref{sec:para}. 
Let $\theta = (\alpha, \rho, \theta_\pi)$ denote the vector of parameters to be estimated, where 
$\alpha$ and $\rho$ are parameters in the utility function $u$ and aggregator function $\phi$ in either the CRRA or CARA Epstein-Zin case. These two parameters capture the agent's risk and intertemporal preferences.  
$\theta_\pi$ represents the parameters in the payoff function $\pi(d_t, x_t, \Delta_t)$.
Throughout the estimation, we assume that the value of the discount factor $\beta$ is known to the econometrician.\footnote{The
identification and estimation of the discount factor $\beta$ in standard time-separable models is known to be difficult and often requires exclusion restrictions or the availability of terminating actions. See discussions in \cite{fang2015estimating}, \cite{bajari2016identification}, \cite{komarova2018joint}, \cite{abbring2020identifying}, \cite{schneider2021identification}.}
Let $x_{it}$ and $d_{it}$ denote the observable state and choice for agent $i=1, 2, \cdots, N$ at period $t=1, 2, \cdots, T$.
$\Delta_{it}$ represents the realization of current period uncertainty. 
We derive the likelihood of the choice pattern and state transition process observed in the data:  
\begin{align}\label{eq:ll}
LL(\theta) = \sum_i^{N} \log\bigg(Pr(d_{i1} | x_{i1}; \theta) \prod_{t=2}^T 
Pr(d_{it} | x_{it}; \theta)Pr(x_{it} |d_{it-1}, x_{it-1}, \Delta_{it-1})
\bigg).
\end{align}
Note that $Pr(d_{it} | x_{it}; \theta)$ represents the choice probability conditional on the observable state $x_{it}$ and $Pr(x_{it} | d_{it-1}, x_{it-1}, \Delta_{it-1})$ represents the state transition probability, which can be estimated separately outside of the structural model.

For a given parameter $\theta$, we solve the ex-ante value function $V(x_t)$ through an iterative procedure outlined as follows. 
\begin{enumerate}[\hspace{0.6cm}(1)]
    \item Initiate the iteration process with a chosen $V^0(x_t)$. Denote the value function at the $r$-th iteration as $V^r(x_t)$.
    \item Given $V^r(x_t)$, compute the choice-specific value function $v(d_t, x_t, \varepsilon^s)$ for a random draw $\varepsilon^s$ from $G$: 
    \begin{equation*}
\begin{split}
\mathbb{E}_{\Delta_{t}|d_{t},x_{t}}\bigg[\phi\bigg(\left(1-\beta\right)u\big(c\left(d_{t},x_{t},\Delta_{t},\varepsilon^s(d_t)\right)\big)+\beta\phi^{-1}\big(\mathbb{E}_{x_{t+1} | d_t, x_t, \Delta_t}[V^r(x_{t+1})]\big)\bigg)\bigg].
\end{split}
\end{equation*}
\item Update $V^{r+1}(x_t)$ by taking the average over $S$ simulation draws:
\begin{align*}
V^{r+1}(x_t) \approx \frac{1}{S}\sum_{s=1}^S\left(\max_{d\in \mathcal{D}}\bigg\{v(d, x_t, \varepsilon^s))  \bigg\} \right). 
\end{align*}
\item Repeat Steps (2)--(3) until convergence.
\end{enumerate}

The contraction property proved in Theorem \ref{theorem:contraction} guarantees the uniqueness of the value function. 
For our specific parameterizations of non-separable preferences, Corollaries \ref{coro:contraction_CRRA} and \ref{coro:contraction} provide sufficient conditions for a contraction mapping when agents have certain preference for the temporal resolution of uncertainty under CRRA and CARA Epstein-Zin preferences, respectively. When the bound $M$ in Theorem \ref{theorem:contraction} is difficult to compute, we can check empirically whether the largest and the smallest fixed points coincide with each other. 
Theorem \ref{theorem:existence} guarantees that we obtain the largest (smallest) fixed point of $T$ through the iterative procedure above if we start iterating from the largest (smallest) possible values of $V$.
In all of our simulation exercises and empirical application, the fixed points are all unique.

At the point of convergence, the conditional choice probabilities $Pr(d_{it}|x_{it};\theta)$ can also be computed numerically.
\begin{align*}
  p(d_t|x_t) \approx \frac{1}{J}\sum_{j=1}^J \mathbf{1}\bigg\{ d_t=\arg\max_{d\in \mathcal{D}}v(d, x_t, \varepsilon^j) \bigg \}.  
\end{align*}
In the outer loop of the algorithm, we search for the value of $\theta \in \Theta$ to maximize the log-likelihood function in Equation (\ref{eq:ll}), i.e., 
the estimator for the structural parameter is 
$\hat{\theta}=\max_{\theta \in \Theta}  LL(\theta)$.

Before moving on to the empirical example, we briefly discuss identification of the model parameters. With the introduction of non-separable time preferences, our model is much richer than the standard additively-separable model. Most importantly, the value functions in our model are non-separable in the random shocks $\varepsilon_t$. 
Existing identification results for time separable models (e.g., \citeauthor{magnac2002identifying}, \citeyear{magnac2002identifying}; \citeauthor{arcidiacono2020identifying}, \citeyear{arcidiacono2020identifying}) rely on the inversion from CCPs to value functions given additively separable random shocks (\citeauthor{hotz1993conditional}, \citeyear{hotz1993conditional}), and thus do not directly apply to our model. 

Our empirical model imposes specific parametric assumptions on $u$ and $\phi$. These parameters affect the agent's risk and intertemporal preferences in non-trivial ways so deriving generic identification results is not a straightforward exercise. All model parameters (i.e., $\theta$) jointly determine the conditional choice probabilities $p(d_t|x_t)$ in Equation \eqref{ccp}.
Matching the model-implied CCPs with their empirical counterparts for various $x_t$, we essentially characterize $\theta$ as the solution to a system of nonlinear equations. Given the high nonlinearity of the dynamic model, one possibility is to at least locally identify $\theta$ given certain rank conditions (see \cite{lewbel2019identification} for more detailed discussions and examples on local identification). Since the main purpose of our paper is to provide a general framework
of dynamic discrete choice models that allow for non-separable time preferences, we
leave a thorough investigation of identification for future research.

\section{Empirical Application: Optimal Engine Replacement Revisited}\label{sec:application}

We apply our model to the dataset of bus engine replacement decisions originally studied in \cite{rust1987optimal}. 
The manager (Harold Zuercher) 
makes a dynamic
choice for each bus engine by trading off between an immediate lump sum cost of replacing it
and higher maintenance costs for keeping it at the beginning of each period.
We extend the Rust model to allow for non-separable time preferences, agents being risk-averse, and earning revenues from operating the bus.

\subsection{Setup}\label{sec:model_emp}

The observable state variable $x_t$ represents the accumulated mileage at the beginning of period $t$. 
Let $\varepsilon_t$ denote the unobserved payoff  shocks and $\Delta_t$ denote the incremental mileage realized within period $t$.
The agent decides whether or not to replace the engine upon observing $(x_t, \varepsilon_t)$.  Let $d_t=1$ represent the case where the engine is replaced, and $d_t=0$ otherwise. The timeline of the model is shown in Figure \ref{fig:timeline1}.

\begin{figure}[htpb!]
    \caption{Timeline of the model}
\begin{center}
\begin{tikzpicture}

  \draw[-stealth] (0,0) -- (14,0);

   \foreach \x in {0,1.5, 4,7,10,12.5}
   \draw (\x cm,4pt) -- (\x cm,-3pt);
   \draw (0,0) node[below=7pt] {period $t$ } node[above=7pt] {};
   \draw (1.5,0) node[below=7pt] {$x_t$} node[above=7pt] {state realizes};
   \draw (4,0) node[below=7pt] {$\ep_t$} node[above=7pt] {shock };
   \draw (7,0) node[below=7pt] {$d_t \in \{0,1\}$} node[above=7pt] {decision};
   \draw (10,0) node[below=7pt] {$\Delta_t$} node[above=7pt] {mileage increases};
   \draw (12.5,0) node[below=7pt] {period $t+1$} ;
   \end{tikzpicture}
\end{center}
    \label{fig:timeline1}
\end{figure}

We first specify the agent's payoff function  
\begin{align}\label{eq:payoff}
    \pi(d_t, x_t, \Delta_t) =
    \begin{cases}
    \theta_d \Delta_t -RC & \text{ if } d_t=1 \\
    \theta_d \Delta_t - \theta_x x_t & \text{ if } d_t=0  \\
    \end{cases}.
\end{align}
In Equation (\ref{eq:payoff}), $RC$ represents the replacement cost and $\theta_d \Delta_t$ represents the revenue collected by the agent within a month for the incremental mileage $\Delta_t$. 
When the bus engine is not replaced, the agent has to pay a maintenance cost proportional to the accumulated mileage, i.e., $\theta_x x_t$. 
The realized consumption agent has in period $t$ is given by
\begin{equation}
c(d_t, x_t, \Delta_t, \varepsilon_t) = \pi(d_t,x_t, \Delta_t) + \sigma \varepsilon_t(d_t),
\end{equation}
where  $\sigma$ is the standard deviation of the random shock and 
$\varepsilon_t(d_t)$ is the shock that the agent receives for alternative $d_t \in \{0, 1\}$. Note that  
$\varepsilon_t=(\varepsilon_t(1), \varepsilon_t(0))$ and it is randomly drawn from a cumulative distribution function $G(\cdot)$.
Note that in this empirical application, the payoff agent receives might be negative. For example, if the agent chooses to replace the engine, he needs to pay an immediate replacement cost. We adopt the 
CARA Epstein-Zin preferences described in Section \ref{sec:cara} as it allows for negative payoffs.

Applying the CARA parameterization to the choice-specific value function defined in Equation (\ref{V2va_simp}), we have
\begin{align*}
v(d_t, x_t, \varepsilon_t) = \frac{1}{\alpha}\bigg[1-\mathbb{E}_{\Delta_t | d_t, x_t}\bigg\{ &(1-\beta)\exp\bigg(-\rho(\pi(d_t, x_t, \Delta_t)+\sigma \varepsilon_t(d_t))\bigg)  \\ 
+&\beta\bigg(1-\alpha V(x_{t+1})\bigg)^{\frac{\rho}{\alpha}}\bigg\} ^{\frac{\alpha}{\rho}}\bigg],
\end{align*}
where $V(x_{t+1})$ is the ex-ante value function given by
\begin{align*}
    V(x_{t+1})=\int \bigg(\max \bigg\{v(0, x_{t+1},\varepsilon), v(1, x_{t+1}, \varepsilon)\bigg\} \bigg)dG(\varepsilon),
\end{align*}
and the next period accumulated mileage $x_{t+1}$ is updated following:
\begin{align*}
    x_{t+1}=(1-d_t)x_t+\Delta_t.
\end{align*}
Intuitively, if the agent chooses to replace the engine, i.e., $d_t=1$, the current accumulated mileage is reset to zero, so the future mileage equals to the incremental mileage $\Delta_t$. If the agent does not replace the engine, the next-period accumulated mileage is increased by $\Delta_t$. Note that in this engine replacement example, the future state variable $x_{t+1}$ is uniquely determined once $d_t$, $x_t$ and $\Delta_t$ are realized. In other words, there is no uncertainty with respect to $x_{t+1}$ conditional on $(d_t, x_t, \Delta_t)$. We therefore omit the expectation $\mathbb{E}_{x_{t+1}|d_t, x_t, \Delta_t}$ in front of $V(x_{t+1})$.
The probability that an engine is replaced conditional on the observed state $x_t$ is
\begin{align*}
p(d_t=1 | x_t)=\int \mathbf{1}\bigg(v(1, x_t, \varepsilon_t) > v(0, x_t, \varepsilon_t) \bigg)dG(\varepsilon_t).
\end{align*}

{
Based on the empirical model described above, we conduct a numerical exercise to compare the conditional choice probabilities under separable and non-separable preferences. 
In Figure \ref{fig:mispec}, the red and blue curve represent the probabilities of replacing the engine when the agent has separable and non-separable time preferences, respectively. Throughout the exercise, we fix the intertemporal preference parameter $\rho$ to be 0.5 and change the value of risk preference parameter $\alpha$. Note that when $\alpha=\rho=0.5$, the two models are equivalent, resulting in the same conditional choice probabilities (i.e., the red and blue curves intersect at $\alpha=0.5$). When the agent prefers late resolution of uncertainty (i.e., $\rho > \alpha$), if we misspecify the agent's preference to be time-separable and use the observed CCP to estimate the agent's risk preference, we would over-estimate $\alpha$. Conversely, we would under-estimate $\alpha$ if the agent prefers early resolution of uncertainty (i.e., $\rho<\alpha$). 
This exercise highlights the consequences of misspecifying the agent's preference to be time-separable when it is in fact not. Importantly, the direction of the bias depends on whether the agent prefers early or late resolution of uncertainty. 

\begin{figure}
 \caption{Misspecifying nonseparable preferences with separable preferences: biased estimates of the risk preference parameter $\alpha$} 
 \label{fig:mispec} 
 \begin{center}\includegraphics[width=0.80\textwidth]{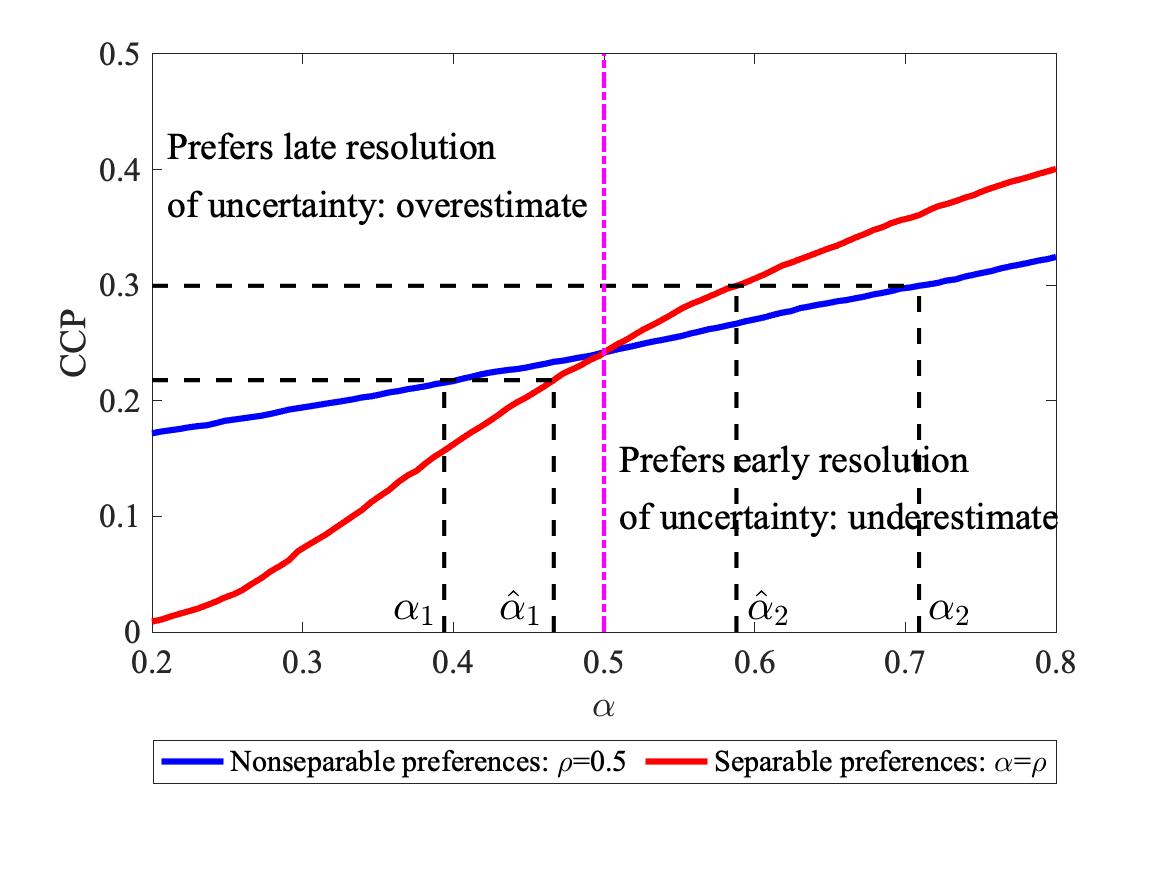}
     \end{center}   
\footnotesize 
\textit{Note}: In this numerical exercise, we set $RC=3$, $\theta_d=3$, $\theta_x=0.5$, $\sigma=2$, $\beta=0.9$, and $\rho=0.5$. $\alpha$ takes values from [0.2, 0.8]. For simplicity, we allow the incremental and accumulated mileages to take 3 values: 0, 1, and 2. The state transition probabilities are specified as: $Pr(\Delta_t=0|x_t=0)=0$, $Pr(\Delta_t=1|x_t=0)=0.5$,
$Pr(\Delta_t=2|x_t=0)=0.5$; $Pr(\Delta_t=0|x_t=1)=0.2$, $Pr(\Delta_t=1|x_t=1)=0.6$,
$Pr(\Delta_t=2|x_t=1)=0.2$;
$Pr(\Delta_t=0|x_t=2)=0.6$, $Pr(\Delta_t=1|x_t=2)=0.4$,
$Pr(\Delta_t=2|x_t=2)=0$.
\end{figure}


\subsection{Estimation Results}\label{sec:est_results}

In the estimation, we discretize mileage into 130 intervals of length 3000 (miles) and assume that $\varepsilon$'s are drawn from a standard normal distribution.
To identify the scale of the unobserved random shock, we fix the $RC$ to be 8 based on the cost of engine replacement reported by Harold Zurcher (see \cite{rust1987optimal} Table III). 
We set the discount factor $\beta=0.9$ and use 2,500 simulation draws when approximating value functions in the iteration process.
To summarize, the structural primitives to be estimated include: parameters for risk-aversion and intertemporal preferences, $\alpha$ and $\rho$; coefficient for revenue and maintenance cost, $\theta_d$ and $\theta_x$; 
and standard deviation of the idiosyncratic shock, $\sigma$.

In Table \ref{tab:main_est_breakdown}, we present the estimation results for four model specifications with standard errors provided in the parenthesis. 
We first estimate a model with non-separable time preferences, where we impose no restrictions on the values of $\alpha$ and $\rho$. Second, we restrict that $\alpha=\rho$, which leads to a model with separable time preferences and risk aversion. Next, we consider adding revenue to the original Rust model. 
Finally, for comparison purposes we also estimate the original Rust model\footnote{In \cite{rust1987optimal}, the distribution of $\varepsilon$ is assumed to be Type I extreme values, mainly for computational convenience. 
We use normal distribution for the shocks $\varepsilon$.}.

Note that in the Rust model with revenue, the choice-specific value function is define as
\begin{equation*}
v(d_t, x_t, \varepsilon_t) =  \mathbb{E}_{\Delta_t | d_t, x_t}\bigg[(1-\beta)\bigg(\pi(d_t, x_t, \Delta_t)+\sigma\varepsilon_t(d_t)\bigg)+\beta V(x_{t+1})\bigg].
\end{equation*}
This is nested in the non-separable case where $\alpha = \rho \rightarrow 0$. 
The original Rust model further removes the revenue from the payoff function $\pi$ (or in other words, restricting $\theta_d=0$), so that the choice-specific value function is reduced to:
\begin{equation*}
v(d_t, x_t, \varepsilon_t) =  (1-\beta)\bigg(\pi(d_t, x_t)+\sigma\varepsilon_t(d_t)\bigg)+\beta \mathbb{E}_{\Delta_t| d_t, x_t}[V(x_{t+1})].
\end{equation*}

From the first column of Table \ref{tab:main_est_breakdown}, we find that $\rho > \alpha$, suggesting that the agent is likely to prefer late resolution of uncertainty. Comparing the log-likelihood between Columns (1) and (2), we are able to reject the hypothesis that the agent has a separable preference at a 10\% significant level (the likelihood ratio test statistic = 2.747, $p$-value = 0.097). 
When allowing for non-separable time preferences, our estimate for $\alpha$ is smaller. This is consistent with our simulation exercises: when agents prefer late resolution of uncertainty, misspecifying the agent's preference to be time-separable results in an overestimate of the risk preference parameter $\alpha$.

Our estimates for the payoff parameters $\theta_d$ and $\theta_x$ are also significantly different when allowing for non-separable preferences. 
In particular, we underestimate the maintenance cost parameter $\theta_x$, while overestimate the revenue parameter $\theta_d$ under time-separable preferences.
This result highlights that misspecifying agents' preferences can lead to biased inferences about structural primitives, and potentially misleading policy recommendations. Comparing Columns (3)-(4), we find that adding revenue to the original rust model has a minimal impact.

\begin{table}[htbp!]
	\centering
	\caption{Comparing estimation results for different models}
	\label{tab:main_est_breakdown}
	\vspace*{0.2cm}
	\scalebox{0.85}{
\begin{tabular}{lcccc} \hline \hline 
& (1) & (2) & (3) & (4) \\
           & Nonseparable & Separable & Rust - rev & Rust-orig \\ \hline 
$\theta_d$ & 0.0526       & 0.1019    & 0.0001     &           \\
           & (0.0073)     & (0.0829)  & (0.0351)   &           \\
$\theta_x$ & 0.1077       & 0.0329    & 0.0208     & 0.0208    \\
           & (0.0122)     & (0.0022)  & (0.0009)   & (0.0014)  \\
$\sigma$   & 1.6070       & 1.5436    & 1.4883     & 1.4770    \\
           & (0.0491)     & (0.0623)  & (0.0542)   & (0.0566)  \\
$\alpha$   & 0.1023       & 0.1457    &            &           \\
           & (0.5087)     & (0.0095)  &            &           \\
$\rho$     & 0.5555       &           &            &           \\
           & (0.0200)     &           &            &           \\ \hline 
LL         & -299.4404    & -300.8139 & -301.4402  & -301.6273 \\ \hline \hline 
\end{tabular}
}
\end{table}

The parameters that capture agents' risk and time preferences ($\alpha$ and $\rho$) are important for predicting choices under uncertainty in dynamic settings and for making policy recommendations. 
To better interpret the estimated values of $\alpha$ and $\rho$, we compute the certainty equivalent given different model specifications. In the engine replacement model, the agent faces uncertainty about the incremental mileage in each period, which not only affects the current-period revenue, but also affects future maintenance costs.
We compute the monetary payoff that makes the agent indifferent when the uncertainty about incremental mileage in each period is removed. Specifically, we assume that in a counterfactual scenario there is a subsidy program that helps the agent to smooth revenue across periods. The agent receives $C$ each period no matter what the incremental mileage is. This is similar to the ``capacity payment" often seen for power
plants in the electricity industry and pipelines in the natural gas industry.\footnote{Regulators have adopted capacity payment mechanisms to encourage capacity investment and promote system reliability. See details of capacity payment mechanisms in electricity markets in \cite{pfeifenberger2009comparison}.} 

We compute the certainty equivalent using estimates from both separable and non-separable models.  
We find that agents with separable preferences are willing to accept a sure payment of \$120 each month to stay indifferent. However, for agents with non-separable time preferences, the calibrated subsidy payment is \$61.6, almost a half of what an agent with separable time preferences would require. This result is intuitive, because when the agent prefers late resolution of uncertainty, the original setting where there is uncertainty in revenue is less unfavorable.
This exercise confirms that misspecifying agents’ preferences as time-separable when it is in fact not can give rise to misleading policy implications.

\section{Conclusion}\label{sec:conclusion}

In this paper, we propose an empirical model of dynamic discrete choice with Epstein-Zin preferences, generalizing the well-known bus engine replacement model studied in \citet{rust1987optimal}. Based on the Bellman equation with Epstein-Zin preferences, we prove the
existence of the value function and provide conditions under which it is the unique fixed point of a contraction mapping. We propose a simulated nested fixed point algorithm for
model estimation based on the theoretical results. Since Epstein-Zin preferences include
 separable expected utility as a special case, our framework allows us to test whether
an agent has separable time preferences or not.

We apply the framework to the bus engine replacement data and show that we can reject the null hypothesis that Harold has time-separable preferences.  
Based on our estimates, the agent is likely to prefer late resolution of uncertainty. Misspecifying the agent's preference to be time separable when it is in fact not leads to a biased inference of agents' risk and intertemporal preferences as well as other key structural parameters. This has non-trivial consequences for making policy evaluations. 

The tools we develop in this paper can be easily applied to analyze many empirical dynamic discrete choice models in industrial organization, environmental, health, and labor economics while allowing for more general non-separable time preferences. 
For example, a recent empirical literature studies the impact of uncertainty on economic outcomes such as oil drilling, firm investment, technology adoption,  and exit decisions (see \citeauthor{kellogg2014effect}, \citeyear{kellogg2014effect};  \citeauthor{dorsey2019waiting}, \citeyear{dorsey2019waiting}; \citeauthor{handley2020measuring}, \citeyear{handley2020measuring}; \citeauthor{gowrisankaran2022policy}, \citeyear{gowrisankaran2012dynamics}).
In these applications, preferences for the temporal resolution of uncertainty 
would naturally play an important role in firms' dynamic responses to policy changes, such as resolving policy uncertainty early. 
Allowing for more general non-separable time preferences could lead to different welfare implications.

\clearpage 
\appendix


\section{Proofs}

To reduce  notational burden, we drop the subscript $t$ in the following proofs. 

\subsection{Proof of Theorem \ref{theorem:existence}}

\label{sec:proof_theorem_1}

Fix some $\underline{v}\leq v_{\ast}$ and $\bar{v}\geq v^{\ast}$.
We first show that for any $v\in\mathcal{V}$, if $v\leq\bar{v}$,
then $T\left(v\right)\leq\bar{v}$. Note that 
\begin{alignat*}{1}
T\left(v\right)\left(x\right) & \leq\int\max_{d\in D}\phi\left(\left(1-\beta\right)u\left(\bar{\pi}+\varepsilon_{d}\right)+\beta\phi^{-1}\left(\mathbb{E}\left[\bar{v}\right]\right)\right)dG\left(\varepsilon\right)\\
 & =\int\phi\left(\left(1-\beta\right)u\left(\bar{\pi}+\max_{d\in D}\varepsilon_{d}\right)+\beta\phi^{-1}\left(\bar{v}\right)\right)dG\left(\varepsilon\right)\\
 & =\mathbb{E}\left[\phi\left(\left(1-\beta\right)u\left(\bar{\pi}+\hat{\varepsilon}\right)+\beta\phi^{-1}\left(\bar{v}\right)\right)\right]
\end{alignat*}
First, if $\phi$ is convex, then
\begin{alignat*}{1}
T\left(v\right)\left(x\right) & \leq\mathbb{E}\left[\left(1-\beta\right)\phi\left(u\left(\bar{\pi}+\hat{\varepsilon}\right)\right)+\beta\bar{v}\right]\\
 & \leq\left(1-\beta\right)\mathbb{E}\left[\phi\left(u\left(\bar{\pi}+\hat{\varepsilon}\right)\right)\right]+\beta\bar{v}\\
 & \leq\left(1-\beta\right)v^{\ast}+\beta\bar{v}\leq\bar{v}
\end{alignat*}
as desired. Now, if $\phi$ is concave, then 
\begin{alignat*}{1}
T\left(v\right)\left(x\right) & \leq\phi\left(\mathbb{E}\left[\left(1-\beta\right)u\left(\bar{\pi}+\hat{\varepsilon}\right)+\beta\phi^{-1}\left(\bar{v}\right)\right]\right)\\
\phi^{-1}\left(T\left(v\right)\left(x\right)\right) & \leq\left(1-\beta\right)\mathbb{E}\left[u\left(\bar{\pi}+\hat{\varepsilon}\right)\right]+\beta\phi^{-1}\left(\bar{v}\right)\\
 & \leq\left(1-\beta\right)\phi^{-1}\left(v^{\ast}\right)+\beta\phi^{-1}\left(\bar{v}\right)\\
 & \leq\phi^{-1}\left(\bar{v}\right)
\end{alignat*}
as desired. 

For the case of $v\geq\underline{v}$ implies $T\left(v\right)\geq\bar{v}$,
note that if $v\geq\underline{v}$, then
\begin{alignat*}{1}
T\left(v\right)\left(x\right) & \geq\int\max_{d\in D}\phi\left(\left(1-\beta\right)u\left(\underline{\pi}+\varepsilon_{d}\right)+\beta\phi^{-1}\left(\mathbb{E}\left[\underline{v}\right]\right)\right)dG\left(\varepsilon\right)\\
 & \geq\mathbb{E}\left[\phi\left(\left(1-\beta\right)u\left(\underline{\pi}+\varepsilon\right)+\beta\phi^{-1}\left(\underline{v}\right)\right)\right]
\end{alignat*}
where the second inequality follows from the fact that $\varepsilon_{d}$
are identically distributed. The rest of the proof follows symmetrically
as the argument above.

Thus, $T:\mathcal{V}\rightarrow\mathcal{V}$ is well-defined. Let
$\geq$ be the partial order on $\mathcal{V}$ where $v\geq\tilde{v}$
if $v\left(x\right)\geq\tilde{v}\left(x\right)$ for all $x\in\mathcal{X}$.
Note that for any subset $\mathcal{U}\subset\mathcal{V}$, we can
define
\[
v_{\wedge}\left(x\right):=\inf_{v\in\mathcal{U}}v\left(x\right)
\]
which is the greatest lower bound of $\mathcal{U}$. Symmetrically,
we can define its least upper bound so $\left(\mathcal{V},\geq\right)$
is a complete lattice. Moreover, it is straightforward to see that
$T$ is monotonic, that is, $v\geq\tilde{v}$ implies $T\left(v\right)\geq T\left(\tilde{v}\right)$.
Thus, by Tarski's fixed point theorem, the set of fixed points of
$T$ is non-empty and also a complete lattice.  

Finally, define $v_{0}=\underline{v}$ and $v_{n}=T\left(v_{n-1}\right)$.
Since $\underline{v}\leq v$ for all $v\in\mathcal{V}$, $v_{0}\leq v_{1}$
so by induction, $v_{n}\leq v_{n+1}$ for all $n$. Since $v_{n}$
is an increasing sequence, it converges to some $v^{\ast}$. Since
$T$ is continuous, we have
\[
v^{\ast}=T\left(v^{\ast}\right)=\lim_{n}T^{n}\left(\underline{v}\right)
\]
To see why and $v^{\ast}$ is the smallest fixed point, suppose there
is some other fixed point $v^{\ast\ast}$. Since $v_{0}\leq v^{\ast\ast}$,
we have 
\[
T^{n}\left(v_{0}\right)\leq T^{n}\left(v^{\ast\ast}\right)=v^{\ast\ast}
\]
so taking limits, we obtain $v^{\ast}\leq v^{\ast\ast}$. The case
for $\bar{v}$ is symmetric.

\subsection{Proof of Theorem \ref{theorem:contraction}}\label{sec:proof_theorem_2}

Recall that
\[
\psi_{y}\left(z\right)=\phi_{y}\left(\phi^{-1}\left(z\right)\right)
\]
so 
\[
T\left(v\right)\left(x\right)=\int\max_{d\in D}\mathbb{E}_{x,d}\left[\psi_{y}\left(\mathbb{E}_{d,x,\Delta}\left[v\right]\right)\right]dG\left(\varepsilon\right)
\]
where $y=u\left(\pi_{x,d,\Delta}+\varepsilon_{d}\right)$ depends
on the realizations of $x,d,\Delta$. We thus have
\begin{alignat*}{1}
 & T\left(v\right)\left(x\right)-T\left(\hat{v}\right)\left(x\right)\\
= & \int\left(\max_{d\in D}\mathbb{E}_{x,d}\left[\psi_{y}\left(\mathbb{E}_{d,x,\Delta}\left[v\right]\right)\right]-\max_{d\in D}\mathbb{E}_{x,d}\left[\psi_{y}\left(\mathbb{E}_{d,x,\Delta}\left[\hat{v}\right]\right)\right]\right)dG\left(\varepsilon\right)\\
\leq & \int\max_{d\in D}\mathbb{E}_{x,d}\left[\psi_{y}\left(\mathbb{E}_{d,x,\Delta}\left[v\right]\right)-\psi_{y}\left(\mathbb{E}_{d,x,\Delta}\left[\hat{v}\right]\right)\right]dG\left(\varepsilon\right)
\end{alignat*}
By the mean value theorem,
\[
\psi_{y}\left(z_{1}\right)-\psi_{y}\left(z_{2}\right)=\psi_{y}^{\prime}\left(z\right)\left(z_{1}-z_{2}\right)
\]
for some $z$ strictly between $z_{1}$ and $z_{2}$. Thus,
\begin{alignat*}{1}
 & T\left(v\right)\left(x\right)-T\left(\hat{v}\right)\left(x\right)\\
\leq & \int\max_{d\in D}\mathbb{E}_{x,d}\left[\sup_{\pi\in\left[\underline{\pi},\bar{\pi}\right],z\in\left[\underline{v},\bar{v}\right]}\psi_{u\left(\pi+\varepsilon_{d}\right)}^{\prime}\left(z\right)\left(\mathbb{E}_{d,x,\Delta}\left[v\right]-\mathbb{E}_{d,x,\Delta}\left[\hat{v}\right]\right)\right]dG\left(\varepsilon\right)\\
= & \int\max_{d\in D}\sup_{\pi\in\left[\underline{\pi},\bar{\pi}\right],z\in\left[\underline{v},\bar{v}\right]}\psi_{u\left(\pi+\varepsilon_{d}\right)}^{\prime}\left(z\right)\mathbb{E}_{x,d}\left[v-\hat{v}\right]dG\left(\varepsilon\right)\\
\leq & \int\max_{d\in D}\sup_{\pi\in\left[\underline{\pi},\bar{\pi}\right],z\in\left[\underline{v},\bar{v}\right]}\psi_{u\left(\pi+\varepsilon_{d}\right)}^{\prime}\left(z\right)dG\left(\varepsilon\right)\left\Vert v-\hat{v}\right\Vert \\
= & M\left\Vert v-\hat{v}\right\Vert 
\end{alignat*}
where
\[
M:=\int\max_{d\in D}\sup_{\pi\in\left[\underline{\pi},\bar{\pi}\right],z\in\left[\underline{v},\bar{v}\right]}\psi_{u\left(\pi+\varepsilon_{d}\right)}^{\prime}\left(z\right)dG\left(\varepsilon\right)
\]
We thus have
\[
\left\Vert T\left(v\right)-T\left(\hat{v}\right)\right\Vert \leq M\left\Vert v-\hat{v}\right\Vert 
\]
If $M<1$, then $T$ is a contraction mapping.

\clearpage 

\clearpage 
\bibliographystyle{ecta}
\bibliography{main}

\end{document}